\documentclass[letterpaper, 10 pt, conference]{ieeeconf}  

\IEEEoverridecommandlockouts                              

\usepackage{graphicx} 
\usepackage{epsfig} 
\usepackage{amsmath} 
\usepackage{amssymb}  
\usepackage{pstool}
\usepackage{xcolor}

\definecolor{Purple}{HTML}{7E2F8E}
\definecolor{DarkBlue}{HTML}{0072BD}

\definecolor{mycolor1}{rgb}{0.4660,0.6740,0.1880}
\definecolor{mycolor2}{rgb}{0.8500,0.3250,0.0980}
\definecolor{mycolor3}{rgb}{0.9290,0.6940,0.1250}
\definecolor{mycolor4}{rgb}{0.4940,0.1840,0.5560}
\definecolor{mycolor5}{rgb}{0.4660,0.6740,0.1880}

\newtheorem{definition}{Definition}
\newtheorem{theorem}{Theorem}
\newtheorem{lemma}{Lemma}

\title{\LARGE \bf
String Stable Integral Control of Vehicle Platoons with Actuator Dynamics and Disturbances
}

\author{Guilherme Fróes Silva, Alejandro Donaire, Aaron McFadyen and Jason Ford
\thanks{G. F. Silva, A. McFadyen and J. Ford are with the School of Electrical Engineering and Robotics, Queensland University of Technology, 2 Geroge St, 4000, QLD, Australia. {\tt\small (g.froessilva, aaron.mcfadyen, j2.ford)@qut.edu.au}}%
\thanks{A. Donaire is with the School of Engineering, University of Newcastle, University Drive, 2308, NSW, Australia. {\tt\small alejandro.donaire@newcastle.edu.au}}%
}

\newcommand{\realnumbers}{\ensuremath{\mathbb{R}}}
\begin{document}

\maketitle
\thispagestyle{empty}
\pagestyle{empty}

\begin{abstract}
This paper presents the design of an integral controller for vehicle platoons with actuator dynamics. The proposed controller ensures string stability with disturbances and simultaneously compensates for constant disturbances through integral action. Sufficient conditions for string stability are satisfied by the use of a suitable state transformation. The proposed controller guarantees disturbance string stability for a prescribed time constant of the actuator dynamics, and we show through simulation that platoons with faster dynamics are also made disturbance string stable.
\end{abstract}

\section{INTRODUCTION}

Systems with multiple agents are advantageous for a wide range of applications, such as general networked systems \cite{Stuedli2017}, cooperative systems \cite{Arcak2007}, and coordination of aerial vehicles \cite{Cai2011}. 
Naturally, controlling a networked system of multiple agents is more challenging than controlling individual agents. Stability, \emph{e.g.} in a Lyapunov sense, is often established for individual agents \cite{Ogata2010}, whereas string stability is desirable when agents are networked in one dimension \cite{Swaroop1996}, and mesh \cite{Pant2002} or swarm \cite{Cai2011} stability is sought for in higher dimensions.

In transportation systems, the action of grouping vehicles into platoons increases traffic throughput \cite{Arem2006} and improves fuel consumption efficiency \cite{Naus2010}. Generally, the vehicles communicate with their neighbours (through local measurements) and, optionally, receive reference information through communication channels. 
In this paper, we show a controller design for bidirectional platoons of heterogeneous vehicles with actuator dynamics that guarantees disturbance string stability whilst rejecting constant disturbances.

Research on vehicle platooning dates back to 1960 \cite{Levine1966}, when Levine and Athans proposed an optimal centralised controller for a string of moving vehicles. Later on, it was observed that disturbances and initial condition perturbations could lead to an effect which would amplify state errors down the string. The property that prevents this effect from happening is called string stability \cite{Peppard1974}. That is, an interconnected system is termed string stable if, and only if, disturbances (and initial condition perturbations) are attenuated from one agent to the other \cite{Swaroop1996,Besselink2017}. 
There are, however, many definitions of string stability in the literature, which depend on the system's communication structure, formation (or spacing policy), and node dynamics (or vehicle dynamics). For an exhaustive review, check \cite{Feng2019} and reference within. 

The communication structure of a platoon can be, for instance,  bidirectional \cite{Knorn2014,Ferguson2017,Herman2017}, where data flows from preceding and following vehicles, or predecessor following \cite{Naus2010,Rogge2008}, where data flows from preceding vehicles only, and possibly reference (leader) information \cite{Monteil2019,Pant2001}. Bidirectional strings can also take information from preceding and following neighbours asymmetrically \cite{Herman2017,Monteil2019}.
The platoon formation is dictated by the spacing policy, be it to maintain constant distance or constant headway time between agents \cite{Darbha1999}. It was shown that, under a constant spacing policy, string stability cannot be achieved with a limited communication range even when using integral action \cite{Yadlapalli2006,Darbha2010}. In
\cite{Seiler2004}, leader position broadcasting avoids disturbance amplification, while using only local relative measurements also leads to string instability for any linear controllers \cite{Barooah2005}. 

Finally, the node dynamics, or vehicle dynamics in platooning literature, can be a second order (double integrator) model \cite{Wit1999,Naus2010}, third order with actuator dynamics \cite{Bian2019,Zhou2005,Shaw2007}, and nonlinear model \cite{Knorn2014, Silva2020, Herman2017,Monteil2019}. The actuator dynamics captures the vehicles' power-train time lag and its impact in stability has been studied in \cite{Bian2019}. Furthermore, a platoon is termed homogeneous if vehicles have equal dynamics and heterogeneous otherwise.

For heterogeneous platoons with nonlinear dynamics and constant spacing policy, a port-Hamiltonian description with integral action addition was proposed to guarantee a weaker form of string stability, coined ``weak $L_2$ string stability'' \cite{Knorn2014}. However, their approach required communication between vehicles, which was later relaxed \cite{Ferguson2017}.
Although still in the ``weak $L_2$ string stability'' setting, it was shown that asymmetry in coupling improves platoon performance \cite{Herman2017}.

Recently, a definition of disturbance string stability was proposed in \cite{Besselink2017}, and sufficient conditions for this properties were given in \cite{Monteil2019}. These sufficient conditions allow for control design of car platoons.
Integral action addition was proposed to incorporate constant disturbance rejection \cite{Silva2020}. 

The contribution of this paper is the design of a string stable integral controller capable of rejecting disturbances for heterogeneous platoons of vehicles using a third-order model with actuator dynamics, following a constant spacing policy. We show that the controller designed using the sufficient conditions from \cite{Monteil2019} and integral action addition \cite{Silva2020} allows string stable control of platoons subject to actuator dynamics.

This paper is organised as follows. In Section \ref{sec:problem}, we present the vehicle platoon dynamics and the control objectives. In Section \ref{sec:control}, we show the control design procedure. Numerical results and simulation studies are presented in Section \ref{sec:results}. Finally, conclusions are drawn in Section \ref{sec:conclusions}.

\section{VEHICLE PLATOONS} \label{sec:problem}
A platoon with $N\geq1$ vehicles can be described by the following set of equations \cite{Zhou2005},
\begin{align}
	\begin{split}\label{eq:model0}
	\dot{q}_i &= v_i \\
	\dot{v}_i &= {m^{-1}_i}f_i + {m_i^{-1}}\bar{d}_i \\
	\dot{f}_i &= -\tau_i^{-1}f_i + \tau_i^{-1}\bar{u}_i,
	\end{split}
\end{align}
for all $i=\{1,\dots,N\}$, where $q_i \in \realnumbers$, $v_i \in \realnumbers$, $m_i \in \realnumbers$ and $f_i \in \realnumbers$ are the position, velocity, mass, and actuator force of the $i$th vehicle, respectively. The acceleration of the vehicles is directly affected by the disturbance $\bar{d}_i \in \realnumbers$, which can be decomposed into time-varying $w_i(t) \in \realnumbers$ and constant $\bar{w}_i \in \realnumbers$ disturbances, such that $\bar{d}_i = w_i(t) + \bar{w}_i$. The time constant $\tau_i \in \realnumbers$ accounts for the power-train time lag of vehicle $i$. Finally, the control input is $\bar{u}_i \in \realnumbers$. This model captures the main dynamics for vehicle platoons and is widely used in the literature (see e.g. \cite{Herman2017,Zhou2005,Bian2019,Darbha2017}). We define the state vector $x_i = [q_i \; v_i \; f_i]^T\in \realnumbers^3$ and also define $x_0$ as the virtual agent (or reference) to be followed. The dynamics \eqref{eq:model0} can be written in compact form
\begin{align}
	\dot{x}_i = \phi_i(x_i) + u_i + d_i
	\label{eq:cloriginal}
\end{align}
where $\phi_i(x_i) = [v_i \; (m_i^{-1}f_i) \; (-\tau_i^{-1}f_i)]^T$, $u_i = [0 \; 0 \; \tau_i^{-1}\bar{u}_i]^T$, and $d_i=[0 \; m^{-1}\bar{d}_i \; 0]^T$.

Interconnected systems of this kind may suffer from an effect called string instability, in which disturbances are amplified along the string. To overcome this effect, we require that the closed-loop system is string stable and the vehicles assume a desired configuration $x_i^\star = [q_0 - \delta_{i,0} \; v_0 \; 0]^T$, where $\delta_{i,0} = \sum_{j=0}^{i-1}\delta_{j+1,j}$ is the distance of vehicle $i$ to the reference position $q_0$, and $\delta_{j+1,j}$ is the desired distance between vehicles. The desired configuration verifies $\dot{x}_i^\star = \phi_i(x_i^\star)$ and it is a solution of the system in the absence of disturbances.

\subsection{String Stability}
We define string stability according to \cite{Besselink2017}, where the disturbance string stability (DSS) definition was proposed.

\begin{definition}[Disturbance String Stability]\label{def:ss}
	The system \eqref{eq:cloriginal} is said to be disturbance string stable if there exists a $\mathcal{KL}$ function $\gamma$ and a $\mathcal{K}$ function $\beta$ such that, for any disturbance $d_i$ and initial conditions, we verify, for all $t>0$,
	\begin{align}
		\label{eq:dss}
		\begin{split}
		\sup_i\left\lvert x_i(t)-x_i^\star(t) \right\rvert_2&\leq\gamma\left(\sup_i\left\lvert x_i(0)-x_i^\star(0) \right\rvert_2,t\right)\\&+\beta\left(\sup_i\left\lVert d_i(t) \right\rVert_{\infty}\right).
		\end{split}
	\end{align}
\end{definition} 

This definition is equivalent to well-known string stability definition by Swaroop and Hedrick \cite{Swaroop1996}, which is expressed in $\epsilon-\delta$ form. 
However, the Definition \ref{def:ss} also accounts for external disturbances acting on the vehicles. 

\subsection{Problem Formulation}
Sufficient conditions for DSS of a general class of systems in closed-loop with static controllers were presented in \cite{Monteil2019}. As discussed in \cite{Silva2020}, these conditions can also be used to ensure DSS of vehicle platoons in closed-loop with dynamic controllers that have the form
\begin{align}
	\begin{split}
	\bar{u}_i &= {h}_{i,i-1}(x_i,x_{i-1}) + \varepsilon_i{h}_{i,i+1}(x_i,x_{i+1}) \\&+ {h}_i^0(x_i,x_0) + k\zeta_i, \label{eq:ac_controller}
	\end{split}\\
	\begin{split}
	\dot{\zeta}_i &= {g}_{i,i-1}(x_i,x_{i-1}) + \varepsilon_i{g}_{i,i+1}(x_i,x_{i+1}) \\&+ {g}_i^0(x_i,x_{0}),
	\end{split}
	\label{eq:integraldynamics}
\end{align}
where the functions ${h}(\cdot) \in \realnumbers$ and ${g}(\cdot) \in \realnumbers$ represent smooth couplings between neighbour vehicles. The integral gain $k \in \realnumbers$ along with the state $\zeta_i \in \realnumbers$, with dynamics \eqref{eq:integraldynamics}, add integral action to the controller \cite{Silva2020}. 

The problem is to find the controller functions $h(\cdot)$ and $g(\cdot)$ that satisfy the sufficient conditions and make the closed-loop system DSS.

\section{CONTROLLER DESIGN} \label{sec:control}
We consider the system \eqref{eq:cloriginal} in closed-loop with the controller \eqref{eq:ac_controller}-\eqref{eq:integraldynamics}. In order to reject constant disturbances, we augment the system with the integral state $\zeta_i \in \realnumbers$,
\begin{equation}
\dot{z}_i = \Phi_i z_i + \rho_i + d_i
\label{eq:claugmented},
\end{equation}
where $z_i = [x_i^T \; \zeta_i]^T$ is the augmented state vector. The control input is $\rho_i = [0 \; 0 \; (\tau_i^{-1}\bar{u}_i) \; u_{\zeta,i}]^T$, where $u_{\zeta,i} = {g}_{i,i-1}(\cdot) + \varepsilon_i{g}_{i,i+1}(\cdot) + {g}_i^0(\cdot)$. The dynamics matrix of the augmented system is
\begin{align}
\Phi_i = \begin{bmatrix}
0 & 1 & 0 & 0 \\ 0 & 0 & m^{-1}_i & 0 \\ 0 & 0 & -\tau_i^{-1} & 0 \\ 0 & 0 & 0 & 0
\end{bmatrix}.
\end{align}

Similar to the approach in \cite{Silva2020}, we use the coordinate change $\xi_i = k^{-1}\zeta_i + \bar{w}_i$ to incorporate the constant disturbance $\bar{w}_i$ into the state vector, which results in the modified augmented system below,
\begin{equation}
\dot{y}_i = \Phi_iy_i + \rho_i + \eta_i
\label{eq:claugmented-coordchange},
\end{equation}
where $y_i = [x_i^T \; \xi_i]^T$, with desired configuration $y_i^\star$, and $\eta_i = [0 \; w_i \; 0 \; 0]^T$ contains only the time-varying disturbance $w_i$. It is useful to write 
$\rho_i = H_{i,i-1} + \varepsilon_iH_{i,i+1} + H^{0}_{i}$ where 
$H_{i,i-1} = [0 \; 0\; h_{i,i-1} \; g_{i,i-1}]^T$, 
$H_{i,i+1} = [0 \; 0 \; h_{i,i+1} \; g_{i,i+1}]^T$, and 
$H^{0}_{i} = [0 \; 0 \; ({h^{0}_{i}} + k\zeta_i) \; {g^{0}_{i}}]^T$. Note that we suppressed the state dependency for simplicity.

The direct application of the sufficient conditions of DSS in \cite{Monteil2019} is generally difficult as the linear matrix inequalities (LMIs) obtained from the sufficient conditions cannot be easily solved due to a lack of structure in the matrices. To overcome this difficulty, we propose a state transformation for the augmented system \eqref{eq:claugmented-coordchange} as follow
\begin{align}\label{eq:transformedstate}
\tilde y_i = T\, y_i 
\end{align}
with 
\begin{align}
T = \begin{bmatrix}
1 & \alpha_1 & 0 & 0 \\ 0 & 1 & \alpha_2 & \alpha_3 \\ 0 & 0 & 1 & \alpha_4 \\ 0 & 0 & 0 & 1
\end{bmatrix}, 
\end{align}
where $\alpha_1$, $\alpha_2$, $\alpha_3$ and $\alpha_4$ are coupling constants that are fundamental for structuring the LMIs, and facilitate finding a solution that verifies the condition of Theorem \ref{thm:conditionsDSS} using optimisation tools. Applying the transformation $T$ and using the new states $\tilde y_i$, we can write the dynamics \eqref{eq:claugmented-coordchange} in the form
\begin{align}
\dot{\tilde{y}}_i =T\Phi_iT^{-1}\tilde{y}_i+\tilde{\rho}_i + \tilde{\eta}_i,
\label{eq:cltransformed}
\end{align}
where $\tilde{\rho}_i = T\rho_i$, and $\tilde{\eta}_i = T\eta_i$, and the transformed desired configuration is $\tilde{y}_i^\star = Ty_i^\star$.

\begin{lemma}
First consider the system \eqref{eq:cloriginal} in closed-loop with controller \eqref{eq:ac_controller}, without integral action, that is $k=0$. In that case, provided the sufficient conditions in \cite{Monteil2019} are satisfied, the system is DSS and the following estimate is true,
\begin{equation}\label{eq:boundcormonteil}
	\begin{split}
	\sup_i\left\lvert x_i(t)-x_i^\star(t) \right\rvert_2 &\leq e^{-\bar{c}^2t}\sup_i\left\lvert x_i(0)-x_i^\star(0) \right\rvert_2 \\&+ \frac{1-e^{-\bar{c}^2t}}{\bar{c}^2}\sup_i\left\lVert d_i(t) \right\rVert_{\infty},
	\end{split}
\end{equation}
where $\bar{c}^2 = c^2-b(1+\max_i\varepsilon_i)$.
\end{lemma}
\begin{proof}
See \cite{Monteil2019}.
\end{proof}

We now present some modified sufficient conditions for DSS as a tool to find controllers that turn the closed-loop system DSS.

\begin{theorem}[Sufficient Conditions for DSS]
	\label{thm:conditionsDSS}
	Consider the system \eqref{eq:claugmented-coordchange} with controller \eqref{eq:ac_controller}-\eqref{eq:integraldynamics}. If the controller functions ${h}(\cdot)$ and ${g}(\cdot)$ are such that the following conditions are satisfied,
	\begin{description}
		\item[C1] $ {H}_{i,i-1}(y^\star_{i},y^\star_{i-1}) = 0$, ${H}_{i,i+1}(y^\star_{i},y^\star_{i+1}) = 0$, and ${H}^{0}_{i}(t,y^\star_i,x_0) = 0$;
		\item[C2] for some $c\neq 0$ and $b>0$ \begin{align}\label{eq:conditionsDSS}
        	\begin{split}
        	&\mu_2\left(
        	J_{i,i}
        	\right)\leq -c^2, \\
        	&\max\left\{\left\lVert
        	J_{i,i-1}
        	\right\rVert_2,\left\lVert
        	J_{i,i+1}
        	\right\rVert_2\right\}\leq b, \\
        	& \text{for all } y_i, y_{i-1}, y_{i+1} \in \mathbb{R}^{4};
        	\end{split}
    	\end{align}
		\item[C3] $
		\varepsilon_i < \frac{c^2}{b}-1
		$,
	\end{description}
	where $\mu_2(A)=\max_i\left(\lambda_i[A]_s\right)$, $[A]_s$ is the symmetric part of $A$, and the elements $J_{i,i}\in \realnumbers^{4\times4}$ and $J_{i,i\pm 1}\in \realnumbers^{4\times4}$ of the Jacobian $J \in \realnumbers^{4N\times4N}$ are 
	\begin{align}
    J_{i,i} &= T \Phi_i T^{-1} + T\dfrac{\partial {\rho}_i}{\partial y_i}T^{-1}, \\
    J_{i,i\pm 1} &= T\dfrac{\partial H_{i\pm 1}}{\partial y_{i\pm 1}}T^{-1}.
    \end{align}
	
	Then,
	
	(i) The system \eqref{eq:claugmented-coordchange} is DSS with,
	\begin{equation}\label{eq:estimatedss}
	\begin{split}
	\sup_i\left\lvert y_i(t)-y_i^\star(t) \right\rvert_2 &\leq Ke^{-\bar{c}^2t}\sup_i\left\lvert y_i(0)-y_i^\star(0) \right\rvert_2 \\&+ K\frac{1-e^{-\bar{c}^2t}}{\bar{c}^2}\sup_i\left\lVert {w}_i(t) \right\rVert_{\infty}.
	\end{split}
	\end{equation}
	
	(ii) The system \eqref{eq:cloriginal} is also DSS with,
	\begin{equation}\label{eq:estimatedss_org}
	\begin{split}
	\sup_i\left\lvert x_i(t)-x_i^\star(t) \right\rvert_2 &\leq Ke^{-\bar{c}^2t}\sup_i\left\lvert x_i(0)-x_i^\star(0) \right\rvert_2 \\&+ Ke^{-\bar{c}^2t}\sup_i\left\lvert \zeta_i(0)+k^{-1}\bar{w}_i \right\rvert_2 \\&+ K\frac{1-e^{-\bar{c}^2t}}{\bar{c}^2}\sup_i\left\lVert {w}_i(t) \right\rVert_{\infty}
	\end{split}
	\end{equation}
	where $\bar{c}^2 = c^2-b(1+\max_i\varepsilon_i)$, $K=\dfrac{\max_i(\sigma_{\max}(T))}{\min_i(\sigma_{\min}(T))}$, and $\sigma_{\min}(A)$ and $\sigma_{\max}(A)$ denote the minimum and maximum singular value of $A$.
\end{theorem}
\begin{proof}	
As shown in the previous section, the dynamics of the system \eqref{eq:claugmented} in closed-loop with the controller \eqref{eq:ac_controller}-\eqref{eq:integraldynamics} can be equivalently written, using the transformation \eqref{eq:transformedstate}, as the dynamics \eqref{eq:cltransformed}.

The application of the sufficient conditions in \cite{Monteil2019} to the system \eqref{eq:cltransformed} can be written as \eqref{eq:conditionsDSS}. These conditions are obtained by computing the Jacobian matrices $J_{i,i}$ and $J_{i,i\pm 1}$ from  \eqref{eq:cltransformed} and writing them in terms of $y_i$. As the conditions C1, C2, and C3 are satisfied then the system \eqref{eq:cltransformed} is DSS, which ensures that the following inequality holds true,
\begin{equation}\label{eq:estimatedss_transf}
	\begin{split}
	\sup_i\left\lvert \tilde{y}_i(t)-\tilde y_i^\star(t) \right\rvert_2 &\leq e^{-\bar{c}^2t}\sup_i\left\lvert \tilde y_i(0)-\tilde y_i^\star(0) \right\rvert_2 \\&+ \frac{1-e^{-\bar{c}^2t}}{\bar{c}^2}\sup_i\left\lVert {\tilde \eta}_i(t) \right\rVert_{\infty}.
	\end{split}
\end{equation}

To prove (i), we define $\Lambda_i \triangleq T^{T}T$ and define $\lambda(\Lambda)=\sigma(T)^2$, where $\lambda(\Lambda)$ is the vector of eigenvalues of $\Lambda$, then use that to obtain the following bounds based on the quadratic form $y_i^T\Lambda y_i$,
\begin{align}
	\begin{split}
	\underline{\sigma}\sup_i\left\lvert y_i(t)-y_i^\star(t) \right\rvert_2 &\leq \sup_i\left\lvert \tilde{y}_i(t)-	\tilde{y}_i^\star(t) \right\rvert_2, \\
	\sup_i\left\lvert \tilde{y}_i(0)-	\tilde{y}_i^\star(0) \right\rvert_2 &\leq \bar{\sigma} \sup_i\left\lvert y_i(0)-y_i^\star(0) \right\rvert_2, \\
	\sup_i\left\lVert \tilde{\eta}_i(t) \right\rVert_{\infty} &\leq \bar{\sigma}\left\lVert \eta_i(t) \right\rVert_{\infty},
	\label{eq:bounds}
	\end{split}
\end{align}
where $\bar{\sigma} = \max_i\{\sigma_{\max}(T)\}$ and $\underline{\sigma} = \min_i\{\sigma_{\min}(T)\}$. 
Using \eqref{eq:bounds} in \eqref{eq:estimatedss_transf}, and as $\sup_i \lVert \eta_i(t) \rVert_{\infty} = \sup_i\lVert w_i(t) \rVert_{\infty} $, we obtain the result \eqref{eq:estimatedss}.

To prove (ii), that is, that system \eqref{eq:cloriginal} is DSS with estimate \eqref{eq:estimatedss_org}, we first note that since $y_i = [x_i^T \; \xi_i]^T$, then we can write
\begin{equation}\label{eq:bound1}
	    \sup_i\left\lvert x_i(t)-x_i^\star(t) \right\rvert_2 \leq \sup_i\left\lvert y_i(t)-y_i^\star(t) \right\rvert_2.
\end{equation}

Then, using \eqref{eq:estimatedss} in \eqref{eq:bound1}, we obtain
\begin{equation}
	\begin{split}
	\sup_i\left\lvert x_i(t)-x_i^\star(t) \right\rvert_2 &\leq \sup_i\left\lvert y_i(t)-y_i^\star(t) \right\rvert_2 \\
	&\leq Ke^{-\bar{c}^2t}\sup_i\left\lvert y_i(0)-y_i^\star(0) \right\rvert_2 \\ &+ K\frac{1-e^{-\bar{c}^2t}}{\bar{c}^2}\sup_i\left\lVert w_i(t) \right\rVert_{\infty}.
	\end{split}
	\label{eq:bound2}
\end{equation}
Also, the triangle inequality allows us to write
\begin{equation}
    \begin{split}
        \sup_i\left\lvert y_i(0) - y_i^\star(0) \right\rvert_2 &\leq \sup_i\left\lvert x_i(0) - x_i^\star(0) \right\rvert_2 \\ &+ \sup_i\left\lvert \xi_i(0) \right\rvert_2,
    \end{split}
    \label{eq:triangleineq}
\end{equation}
which can be used in \eqref{eq:bound2} together with $\xi(0) = \zeta_i(0) + k^{-1}\bar{w}_i$ to finally obtain \eqref{eq:estimatedss_org}, which completes the proof.
\end{proof}

It is important to note that to use Theorem \ref{thm:conditionsDSS} we need only to compute the matrices $J_{i,i}$, $J_{i,i\pm 1}$, and $T$, and find functions $h(\cdot)$ and $g(\cdot)$ that satisfy the sufficient conditions C1, C2, and C3. Then, the system \eqref{eq:cloriginal} in closed-loop with the  controller \eqref{eq:ac_controller}-\eqref{eq:integraldynamics} is DSS and also rejects constant disturbances.

\section{SIMULATION STUDIES} \label{sec:results}
We consider the system \eqref{eq:cloriginal} in closed-loop with the controller  \eqref{eq:ac_controller}-\eqref{eq:integraldynamics}, with
\begin{align}
\begin{split}
{h}_{i,i-1} &= h^p_i(q_{i-1}-q_i-\delta_{i,i-1}) + K^v_i(\dot{q}_{i-1}-\dot{q}_i), \\
{h}_{i,i+1} &= h^p_i(q_{i+1}-q_i+\delta_{i+1,i}) + K^v_i(\dot{q}_{i+1}-\dot{q}_i), \\
{h}^{0}_{i} &= K^{p0}_{i}(q_{0}-q_i-\delta_{i,0}) + K^{v0}_{i}(\dot{q}_0-\dot{q}_i),
\end{split}
\end{align}
and 
\begin{align}
\begin{split}
{g}_{i,i-1} &= g^p_{i}(q_{i-1}-q_i-\delta_{i,i-1}) + G^v_{i}(\dot{q}_{i-1}-\dot{q}_i),\\
{g}_{i,i+1} &= g^p_{i}(q_{i+1}-q_i+\delta_{i+1,i}) + G^v_{i}(\dot{q}_{i+1}-\dot{q}_i),\\
{g}^{0}_{i} &= G^{q0}_{i}({q}_0-{q}_i-\delta_{i,0}) + G^{v0}_{i}(\dot{q}_0-\dot{q}_i).
\end{split}
\end{align}

The structure of the controller requires absolute position and velocity, relative position and velocity, and reference information to compute the control input for each agent.

If the nonlinear functions $h^p_i(\cdot)$ and $g^p_i(\cdot)$ have lower and upper bounds, so does the Jacobian as it depends linearly on their partial derivatives \cite{Monteil2019}. The Jacobian $J$ upper and lower bounds are $J_L$ and $J_U$ respectively, leaving us with the following LMIs, from C2 and C3,
\begin{gather*}
[J_{i,i,L}]_s \leq -c^2I_4, \qquad
[J_{i,i,U}]_s \leq -c^2I_4, \\	
\begin{bmatrix}
bI_4 & J_{i,i-1,L} \\ J^T_{i,i-1,L} & bI_4
\end{bmatrix} \geq 0, \quad
\begin{bmatrix}
bI_4 & J_{i,i-1,U} \\ J^T_{i,i-1,U} & bI_4
\end{bmatrix} \geq 0, \\
\begin{bmatrix}
bI_4 & J_{i,i+1,L} \\ J^T_{i,i+1,L} & bI_4
\end{bmatrix} \geq 0, \quad
\begin{bmatrix}
bI_4 & J_{i,i+1,U} \\ J^T_{i,i+1,U} & bI_4
\end{bmatrix} \geq 0.
\end{gather*} where $I_4$ is the 4-by-4 identity matrix and $[J_{i,i,L}]_s \leq -c^2I_4$ is equivalent to $\mu_2(J) \leq c^2$ \cite{Monteil2019}. 

It is guaranteed, through Theorem \ref{thm:conditionsDSS}, that the system \eqref{eq:cloriginal} will be DSS under the nonlinear controller \eqref{eq:ac_controller}-\eqref{eq:integraldynamics} with control gains that satisfy the LMIs above.

We used CVX, a package for specifying and solving convex programs \cite{cvx}, to find a controller by posing an optimisation problem minimising $-\bar{c}^2$, with the LMIs above as constraints. To solve the LMIs, we set $\varepsilon_i = 1$, $\tau_i = 1 \text{ s}$, $m_i = 1$ for all $i = \{1,\dots,N\}$, and the coupling constants  
$\alpha_1=0.8$, 
$\alpha_2=1$, 
$\alpha_3=0.7$, 
and $\alpha_4=-0.5$. We found a set of controller gains that satisfy the conditions C1, C2 and C3, where 
$K^p_{i1} = K^p_{i2} = 0.001$, 
$K^v_i	= 0.001$, 
$K^{p0}_{i}= 0.4631$, 
$K^{v0}_{i}= 0.7$, 
$k   = 0.1436$, 
$G^p_{i1} = G^p_{i2} = 0.001$, 
$G^v_{i} = 0.001$, 
$G^{p0}_{i}= 0.1430$, and 
$G^{v0}_{i}= 0.3082$.

With the system \eqref{eq:cloriginal} in closed-loop with controller \eqref{eq:ac_controller}-\eqref{eq:integraldynamics}, we ran an exhaustive number of simulation studies for the controller proposed in this paper and the controller obtained using \cite[Corollary 1]{Monteil2019}, denoted $C_1$ and $C_2$, respectively.

We ran the simulations for platoons with lengths 
$N=[50,150,\dots,500]$, with initial conditions $x_i(0) = [(q_0(0)-\delta_{i,0}+\Gamma_i) \; (\dot{q}_0(0) + \Gamma_i) \; 0]^T$, inter-vehicle spacing $\delta_{i,i-1} = \delta_{i,i+1} = 10\text{ m}$ and reference speed $\dot{q}_0 = 20\text{ m/s}$. The disturbance is decomposed into time-variant $w_i(t) = \Gamma_i\sin\exp(-0.1t) \text{ m/s}^2$ and constant $\bar{w}_i = (1+\Gamma_i) \text{ m/s}^2$ disturbances, where $\Gamma_i$ is uniformly randomly generated in the interval $[0,1]$.

\subsection{Actuator dynamics as designed}
Let us consider the case where all vehicles' actuator dynamics have time constants equal to the time constant for which the controller $C_1$ was designed, that is $\tau_i = 1$ for all $i = \{1,\dots,N\}$.

Fig. \ref{fig:normstimetconst1} shows the bounds generated by Theorem \ref{thm:conditionsDSS}, that is bound \eqref{eq:estimatedss_org}, the bounds obtained with \cite[Corollary 1]{Monteil2019}, that is bound \eqref{eq:boundcormonteil}, and the state errors obtained with their respective controllers, $C_1$ and $C_2$, for a platoon with $N=500$ vehicles. We note that the state error of the platoon under $C_1$ converges to zero, showing that the closed-loop is DSS whilst constant disturbances are rejected. 
Also, Fig. \ref{fig:normsN} shows that the supreme state error norm of all platoons is not affected by the string length and the closed-loop system is always DSS.

\begin{figure}[tb]
	\centering
	\psfragfig{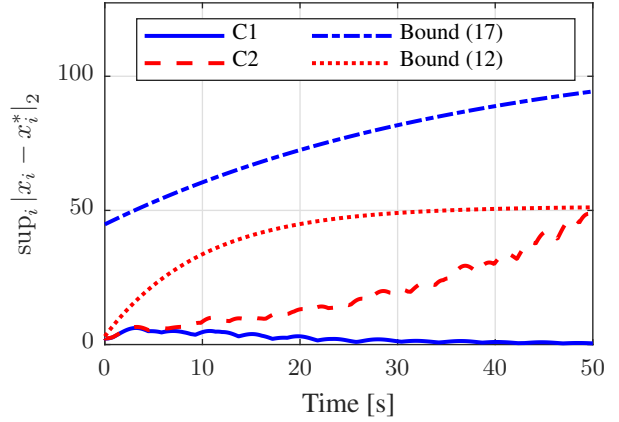}
	\caption{Supreme $L_2$-norm state error for platoon with $N=500$ vehicles and $\tau_i = 1$, for both controllers $C_1$ ({\textcolor{blue}{---}}) and $C_2$ ({\textcolor{red}{-- --}}), and their respective bounds \eqref{eq:estimatedss_org} ({$ \textcolor{blue}{\cdot~-}$}) and \eqref{eq:boundcormonteil} ({$ \textcolor{red}{\cdots}$}).}
	\label{fig:normstimetconst1}
\end{figure}

\begin{figure}[tb]
	\centering
	\psfragfig{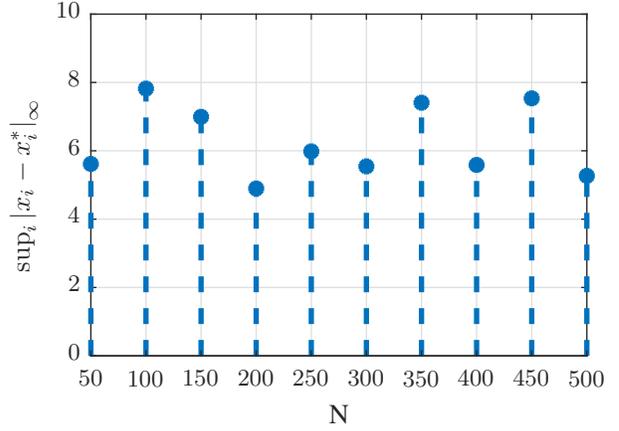}
	\caption{Supreme infinity-norm state error of platoons of different lengths $N$ under controller $C_1$.}
	\label{fig:normsN}
\end{figure}

To facilitate visualisation, we show the states time histories of vehicles $i={2, 100, 250, 400, 500}$ of the platoon of 500 vehicles. Fig. \ref{fig:displacements10} shows the displacement between the vehicle $i=1$ and the reference vehicle $x_0$, that is $e_{i,i-1}=q_{i-1}-q_{i}-\delta_{i,i-1}$, for both controllers $C_1$ and $C_2$, where we notice that only controller $C_1$ compensates for the constant disturbances whilst reducing oscillations caused by the actuator dynamics. Fig. \ref{fig:displacementsii-1} shows the displacement of all other vehicles and their predecessor neighbours where it is again possible to see the constant disturbances being rejected by controller $C_1$. In Fig. \ref{fig:vel_accel}, we show that, with controller $C_1$, good performance is achieved for both velocity and actuator force states, with reasonable force values. The integral states and control inputs are shown in Fig. \ref{fig:integralstate_control}, where we see the integral state converge to a value proportional to the constant disturbance with the control input sustaining reasonable values.

\begin{figure}[tb]
	\centering
	\psfragfig{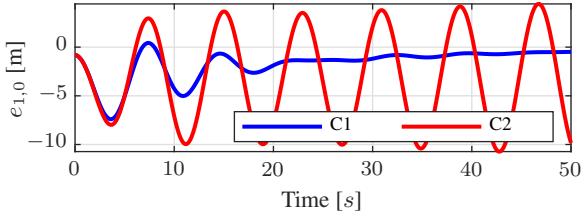}
	\caption{Displacement of vehicle $i=1$ to the reference $x_0$ for both controllers $C_1$ (blue, continuous line) and $C_2$ (red, dashed line), for the platoon with $N=500$ vehicles.}
	\label{fig:displacements10}
\end{figure}

\begin{figure}[tb]
	\centering
	\psfragfig{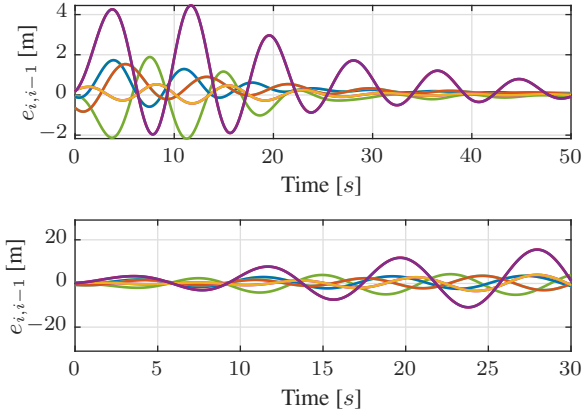}
	\caption{Displacement of vehicles 2 (\textcolor{mycolor1}{\bf--}), 100 (\textcolor{mycolor2}{\bf--}), 250 (\textcolor{mycolor3}{\bf--}), 400 (\textcolor{mycolor4}{\bf--}), and 500 (\textcolor{mycolor5}{\bf--}) to their predecessors, for controller $C_1$ (top) and $C_2$ (bottom), for the platoon with $N=500$ vehicles.}
	\label{fig:displacementsii-1}
\end{figure}

\begin{figure}[tb]
	\centering
	\psfragfig{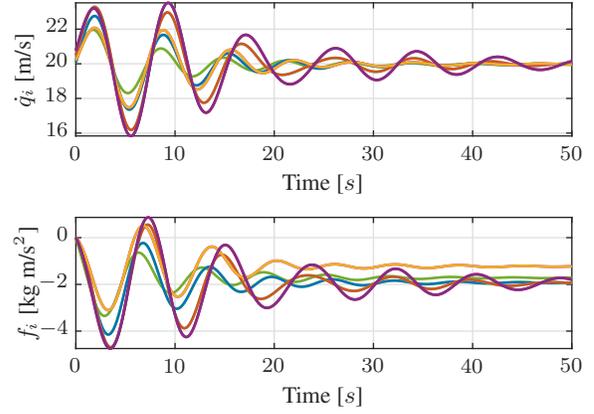}
	\caption{Velocities and actuator forces of vehicles 2 (\textcolor{mycolor1}{\bf--}), 100 (\textcolor{mycolor2}{\bf--}), 250 (\textcolor{mycolor3}{\bf--}), 400 (\textcolor{mycolor4}{\bf--}), and 500 (\textcolor{mycolor5}{\bf--}) for the platoon with $N=500$ under controller $C_1$.}
	\label{fig:vel_accel}
\end{figure}

\setlength{\textfloatsep}{16pt}

\begin{figure}[tb]
	\centering
	\psfragfig{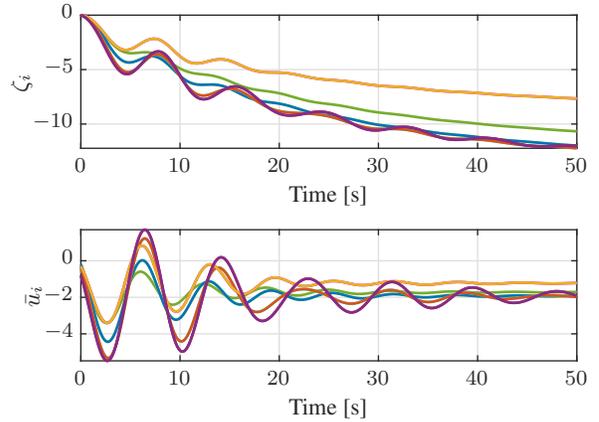}
	\caption{Integral state and control input of vehicles 2 (\textcolor{mycolor1}{\bf--}), 100 (\textcolor{mycolor2}{\bf--}), 250 (\textcolor{mycolor3}{\bf--}), 400 (\textcolor{mycolor4}{\bf--}), and 500 (\textcolor{mycolor5}{\bf--}) for the platoon with $N=500$ under controller $C_1$.}
	\label{fig:integralstate_control}
\end{figure}

\subsection{Actuator dynamics different than designed}
In this section, we simulate platoon systems with actuators dynamics' time constants different than the one for which the controller was designed. Exhaustive simulation studies indicate that when all vehicles' dynamics have time constants less than or equal to the time constant for which the control was designed, the condition C2 is satisfied and the closed-loop system is DSS. However, the analytical proof of that is pending. 

In Fig. \ref{fig:normstime} we show the state errors and bounds for a platoon of length $N=500$ and $\tau_i = 0.5 (1.1 - \Gamma_i)$. In this study, we observe that both controllers make the closed-loop DSS, but $C_1$ is also able to compensate for the constant disturbance.

\begin{figure}[tb]
	\centering
	\psfragfig{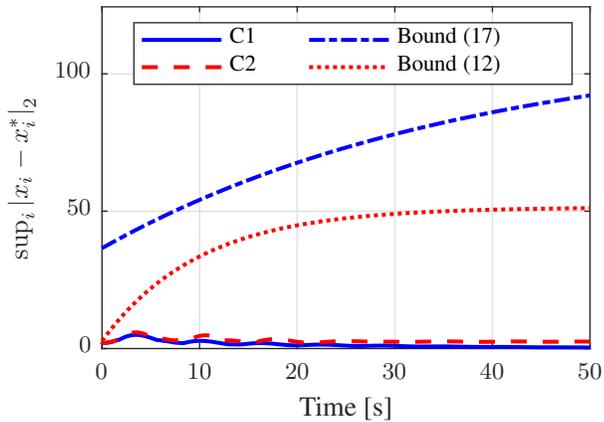}
	\caption{Supreme $L_2$-norm state error for platoon with $N=500$ vehicles and $\tau_i = 0.5 (1.1 - \Gamma_i)$, for both controllers $C_1$ ({\textcolor{blue}{---}}) and $C_2$ ({\textcolor{red}{-- --}}), and their respective bounds \eqref{eq:estimatedss} ({$ \textcolor{blue}{\cdot~-}$}) and \eqref{eq:boundcormonteil} ({$ \textcolor{red}{\cdots}$}).}
	\label{fig:normstime}
\end{figure}

Finally, Fig. \ref{fig:normstime_unstable} shows the state errors and bounds for the platoon with $N=500$ vehicles and $\tau_i = 1.5 (1.1 - \Gamma_i)$, a time constant bigger than the one for which the controller was designed. We observe the condition C2 is not satisfied and the system is unstable.

\begin{figure}[tb]
	\centering
	\psfragfig{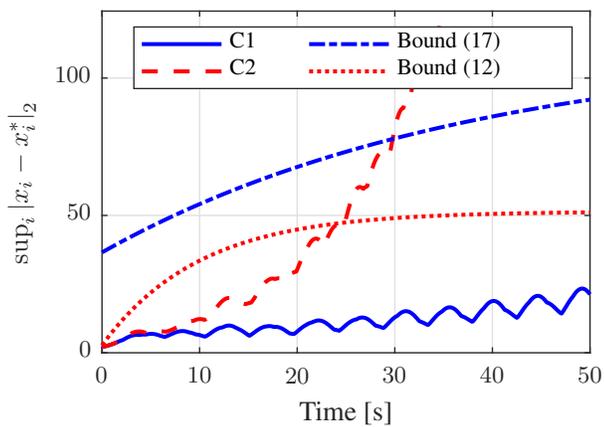}
	\caption{Supreme infinity-norm state error for platoon with $N=500$ vehicles and $\tau_i = 1.5 (1.1 - \Gamma_i)$, for both controllers $C_1$ ({\textcolor{blue}{---}}) and $C_2$ ({\textcolor{red}{-- --}}), and their respective bounds \eqref{eq:estimatedss_org} ({$ \textcolor{blue}{\cdot~-}$}) and \eqref{eq:boundcormonteil} ({$ \textcolor{red}{\cdots}$}).}
	\vspace{-0.2cm}
	\label{fig:normstime_unstable}
\end{figure}

\section{CONCLUSIONS} \label{sec:conclusions}

In this paper we designed a string stable integral controller capable of rejecting disturbances in bidirectional platoons of heterogeneous vehicles with distinct actuator dynamics. During control design, we prescribed the actuator dynamics time constant so the controller guarantees disturbance string stability of the platoon, regardless of its length. 
The controller gains can be computed via offline optimisation before the implementation.
Simulation studies show satisfactory performance of the control system and indicates that disturbance string stability is ensured provided that the actuator dynamics have time constants less than the prescribed value. 

\section{ACKNOWLEDGEMENTS}

G.F.S., A.M. and J.F. acknowledge continued support from the Queensland University of Technology (QUT) through the Centre for Robotics.



\bibliographystyle{IEEEtran}


\bibliography{refcdc2020}

\end{document}